\documentclass[11pt]{article}

\usepackage{abbrevs}
\usepackage{amsthm}
\usepackage{amsmath}
\usepackage{color}
\usepackage{amssymb}
\usepackage{enumerate}

\newtheorem{theorem}{Theorem}
\newtheorem{corollary}{Corollary}

\newtheorem{problem}{Problem}
\newtheorem{conj}{Conjecture}

\newname\stableset{{\rm \textsc{independent-set}}}
\newname\clique{{\rm \textsc{clique}}}
\newname\mmatrix{{\rm \textsc{polytope-M-matrix}}}
\newname\posdet{{\rm \textsc{positive-determinant}}}
\newname\minradius{{\rm \textsc{polytope-min-radius}}}
\newname\maxradius{{\rm \textsc{polytope-max-radius}}}
\newname\polytopicinstability{{\rm \textsc{polytopic-instability}}}

\def\eq#1{(\ref{eq-#1})}
\newcommand{\simplex}{\Delta}
\newcommand{\B}{B_\pi}

\newcommand{\G}{\mathcal{G}}

\newcommand{\re}{{\mathbb R}}

\newcommand{\n}{{\mathbb N}}

\newcommand{\conv}{{{\rm Conv} }}

\providecommand{\com[2]}{\begin{tt}[#1: #2]\end{tt}}

\providecommand{\rmjn[1]}{{#1}}
\providecommand{\rmj[1]}{{#1}}
\providecommand{\ad[1]}{{\color{red}#1}}

\title{Polytopic uncertainty for linear systems: \\ New and old complexity results\thanks{To appear in Systems \& Control Letters, 2014}}

\author{Nikos Vlassis$^1$, Rapha{\"e}l M.\ Jungers$^2$ \vspace*{5pt} \\
\small $^1$Luxembourg Centre for Systems Biomedicine, University of Luxembourg \\
\small $^2$ICTEAM institute, Universit{\'e} Catholique de Louvain}

\begin{document}
\maketitle

\begin{abstract}
We survey the problem of deciding the stability or stabilizability of uncertain linear systems whose region of uncertainty is a polytope.  This natural setting has applications in many fields of applied science, from Control Theory to Systems Engineering to Biology.
We focus on the algorithmic decidability of this property when one is given a particular polytope.  This setting gives rise to several different algorithmic questions, depending on the nature of time (discrete/continuous), the property asked (stability/stabilizability), or the type of uncertainty (fixed/switching).  Several of these questions have been answered in the literature in the last thirty years.  We point out the ones that have remained open, and we answer all of them, except one which we raise as an open question.  In all the cases, the results are negative in the sense that the questions are NP-hard.
As a byproduct, we obtain complexity results for several other matrix problems in Systems and Control.
\end{abstract}

\section{Introduction}

Robust control is an important topic that has motivated several important research lines in Systems and Control since the eighties.  It has spanned a wide range of applications, and has benefited from many different techniques in applied mathematics, such as algorithmic complexity, convex optimization, game theory, $\mu$-analysis, and others (see, e.g., \cite{Doyle11,iglesias2010control}).

One of the simplest settings in robust control is a discrete-time (resp. continuous-time) linear system describing the evolution of a state space vector $x\in \re^n$ as follows:
\begin{align}
x(t+1) &= A(t) \, x(t), \label{eq-discrete} \\
\dot x (t) &= A(t) \, x(t), \label{eq-continuous}
\end{align}
where the matrix $A(t)$ is restricted to belong to a given polytope{\rmjn{\footnote{\rmjn{Note that polytopes are one of the simplest representations of compact sets, and thus the negative results presented in this paper are generalizable to more complex compact sets.}}} $P\subset \re^{n\times n}$. Depending on the context, several different questions may be relevant: In some situations, the matrix is fixed for the whole trajectory, but its actual value is not determined, except for the fact that it belongs to the polytope~$P$. We refer to this case as the \emph{fixed uncertainty} case.  In other situations the matrix $A(t)$ is allowed to change from time to time; the trajectory of the system is determined by a \emph{switching signal} $\sigma:$
\begin{equation}
\sigma:\n\rightarrow P :\quad t\rightarrow A(t),
\end{equation} 
or, in the continuous-time case,
\begin{equation}
\sigma:\re_+\rightarrow P :\quad t\rightarrow A(t).
\end{equation} 
We refer to this case as the \emph{switching uncertainty} case.
Typical questions that one would like to answer are the following:

\begin{problem}{\tt[Stabilizability]}
Given a set of matrices $\{A_1,\dots,A_m\}$ describing a polytope $$P=\conv{\{A_1,\dots,A_m\}},$$ does there exist a matrix $A\in P$ (resp.\ a switching signal $\sigma$) such that the trajectory converges to zero for any initial condition?
\end{problem}
\begin{problem}{\tt[Stability]}
Given a set of matrices $\{A_1,\dots,A_m\}$ describing a polytope $$P=\conv{\{A_1,\dots,A_m\}},$$ does the trajectory converge to zero for every possible matrix $A\in P$ (resp.\ every switching signal $\sigma$) and every initial condition?
\end{problem}

Thus, the above setting raises eight different algorithmic questions, depending on the discrete/continuous nature of time, the stability/stabilizability question, and the fixed/switching nature of uncertainty.
These questions are typical of \emph{Robust Control,} where one approximatively knows the dynamical system, and wants to ensure that it is stable, up to a certain perturbation of the model.  Many situations in practical applications boil down to one of these cases  (see, e.g.,  \cite{Lin09}). 

A case of particular interest is the analysis of {\em Biological systems} \cite{doyle2008robust,he2005global,iglesias2010control}. There, robustness can be needed because the system is a linearization of a real nonlinear system, and one must then take into account the discrepancy between the model and the real-life situation, or it can be required because the system is not perfectly known. Robustness also arises naturally in neural networks, for uncertainty reasons, or because of the switching nature of their dynamics.
The stabilizability problem also turns out to be relevant in situations where one can control a system by activating a certain mode, among a few ones that are available, in order to stabilize the system.  For instance, in \emph{virology}, it has been reported that the drug treatment for some viral disease like HIV could be improved by switching among several medications from time to time \cite{drugs}. Recently, researchers in the Control community have modeled this situation as a switching system, where the vector $x(t)$ represents the different concentrations of viral populations in the blood, and the switching signal corresponds to the choice of medication at every particular time \cite{hmcb10,hcm13}. Typically, in this application, one's goal is to design this switching signal so as to best control the population of viruses in the patient's body.  This situation falls into the scope of the present paper, as it can be modeled as a \emph{continuous-time stabilizability problem with switching uncertainty}. 
Other applications can be found in fields as diverse as \emph{wireless control networks \cite{alurdjpw,ddj12}},  \emph{fluid dynamics} \cite{barbu1998h}, or even \emph{medicine} \cite{GriEtal:2001:IFA_921}.

In the present paper we focus on the algorithmic problem of answering the above questions. It turns out that all the cases that are known are NP-hard.  However, it is important to mention that practical algorithms have been developed for some of these problems, which try to circumvent these negative results, and which provide practitioners with workable methods that often allow to reach satisfactory answers in practice.  For instance, sufficient conditions have been proposed that are \rmjn{efficiently checkable (e.g. in polynomial time)} thanks to modern tools like convex optimization methods for proving stability, or instability, of diverse types of uncertain sets of dynamical systems.  We refer \rmjn{to \cite{protasov-jungers-blondel09,jungers_lncis,protasov-jungers-cdc13,daafouzbernussou,parrilo-jadbabaie,guglielmi2011fast,lee-dullerud06,GP11,margaliot-survey}} and references therein for various examples of such methods. While such methods are of course important because they allow in favorable situations to obtain a solution to the problem, the negative results we present in this paper are also very  valuable, as they allow to understand theoretical barriers that no algorithm could overcome (unless, of course, $P=NP$). Such negative results can also assist in the development of new methods, as they help to understand what type of performance one can hope for in any candidate new method.  

In Table \ref{table-recap} we summarize all known results in the literature (including the ones we derive in this paper).  One can see that all the cases are now known to be NP-hard, except for the stabilizability of continuous-time systems with switching uncertainty.  In Section \ref{sec-scs}, we state a conjecture which (if answered positively) would solve the problem.
\begin{figure}
\begin{tabular}{|l|c|c|c|c|}
   \hline
      & \multicolumn{2}{c|}{Stabilizability} & \multicolumn{2}{c|}{Stability} \\
   \hline
    & Switched & Fixed &Switched & Fixed  \\
   \hline
   Continuous-time & ? & This paper& \cite{gurvits2009np} & \cite{gurvits2009np}\\
   \hline
   Discrete-time & \cite{blondel-mortal}\footnotemark & This paper & \cite{tsitsiklis97lyapunov} & This paper \\
   \hline
\end{tabular}\caption{Summary of the known results.}  \label{table-recap}
\end{figure}
\footnotetext{We note that in \cite{blondel-mortal}, a slightly different problem is analyzed, in which the matrix $A(t)$ is restricted to be a vertex of the polytope $P.$ \rmjn{The problem with arbitrary switching in the entire polytope certainly makes sense for this particular question too, but we are not aware of any analysis of this latter problem in the literature.}}

\section{Discrete-time systems with fixed uncertainty}

In this section we analyze the stabilizability and stability problems for discrete-time systems with fixed uncertainty. These translate to the problems of asking, for a given matrix polytope, whether there exists a matrix in the polytope with spectral radius smaller than one, or larger than one, respectively. Both problems turn out to be NP-hard.

\subsection{Stabilizability}

\begin{problem}[\minradius]
Given a set of real matrices, is there a convex combination of those whose spectral radius is smaller than one? 
\end{problem}

\begin{theorem}
The \minradius problem is NP-hard.
\label{minr}
\end{theorem}

\begin{proof}
We establish a polynomial-time reduction from the \stableset problem. This problem asks, for a given undirected graph $\G=(V,E)$ and a positive integer $j \leq |V|$, whether $\G$ contains an independent set $V'$ (a set of pairwise non-adjacent vertices) with size $|V'| \geq j$. This problem is NP-complete~\cite{Garey79}. \rmjn{We assume $j \geq 2$ (otherwise the problem is trivial)}.

An instance of \minradius takes as input $k$ real $n \times n$ matrices $A_i$, with $i=1,\ldots,k$, and asks whether there exists a nonnegative vector $\pi=(\pi_1,\ldots,\pi_k)^\top$ with $\sum_{i=1}^k \pi_i = 1$, such that the matrix $\B = \sum_{i=1}^k \pi_i A_i$ has spectral radius less than one. We will prove NP-hardness of the problem for the special case in which $k=n$, that is, when the number of input matrices equals their dimension. For any input instance of \stableset, our reduction will construct (in a number of steps at most polynomial in the size of the problem input) an instance of \minradius, such that a polynomial-time algorithm for deciding the latter would imply a polynomial-time for deciding the former.

Let $(\G,j)$ be an instance of \stableset, and let $C$ be the $n \times n$ adjacency matrix of the graph $\G$.
The matrix $C$ is a symmetric zero-one matrix with zeros in the main diagonal. Let $c_i$ denote the $i$'th column of $C$, and let $e_i$ denote the length-$n$ vector with 1 in the $i$'th entry and all other entries zero. The reduction constructs $n$ nonnegative block matrices $A_i$, for $i=1,\ldots,n$, of size $(n+1) \times (n+1)$:
\begin{equation}
A_i = 
 \begin{bmatrix} 
  \emptyset & e_i+c_i \\
  e_i^\top  & \rmjn{r}
 \end{bmatrix},
 \label{eq-A1}
\end{equation} 
where \rmjn{$r=1-\frac{1}{j-1} \in [0,1)$} and $\emptyset$ denotes the $n \times n$ zero matrix. The matrix $\B = \sum_{i=1}^k \pi_i A_i$ then reads
 \begin{equation}
 \B = 
 \begin{bmatrix} 
  \emptyset & (I + C)\pi \\
  \pi^\top  & \rmjn{r}
 \end{bmatrix}.
 \label{eq-B}
\end{equation} 
The special block structure of $\B$ allows to analytically compute its eigenvalues by manipulating the system $\B v = \lambda v$, and it is easy to verify that the spectrum of $\B$ is given by 
\begin{equation}
	\sigma(\B)= \bigg\{ 0, \frac{\rmjn{r} \pm \sqrt{\rmjn{r}^2+4\pi^T(I + C)\pi}}{2}\bigg\}.
\end{equation}
It follows that the spectral radius condition $\max(|\rmjn{\sigma(\B)}|) < 1$ \rmj{is equivalent to the condition} \rmjn{$\pi^\top (I+C) \pi < 1-r$}. Moreover, for any graph $\G$ with adjacency matrix $C$, the following holds \cite{Motzkin65}:
\begin{equation}
  \frac{1}{\alpha(\G)} = \min_{y \in \simplex} \, y^\top (I + C) \, y \, ,
\label{eq-motzkin1}
\end{equation}
where $\alpha(\G)$ is the size of the maximum independent set of $\G$.
Hence, \rmj{there exists a vector $\pi \in \simplex$ that satisfies \rmjn{$\pi^\top (I+C) \pi < 1-r$} if and only if 
$\frac{1}{\alpha(\G)} < \frac{1}{\rmjn{j-1}}$, or $\alpha(\G) > \rmjn{j-1}$.  Thus, there exists an independent set $V' \subseteq V$ with size $|V'| \geq j$ if and only if our set of matrices is stable.} 
This completes the reduction.
\end{proof}

This result establishes that the problem of globally minimizing spectral radius over a matrix polytope is NP-hard (\rmjn{which was independently established recently using a different proof} \cite{Fercoq11}). The spectral radius of a nonsymmetric real matrix affine function such as $B_\pi$ in \eq{B} is generally a nonconvex function, and the problem of globally optimizing the spectral radius is a difficult one \cite{Overton88,Han99,Neumann07}. The above result provides evidence for this `difficulty'. We note that the problem of globally minimizing the spectral radius becomes tractable in the case of symmetric matrices \cite{Overton88,Fiedler92}, or in the case of irreducible nonnegative matrices whose entries are {\em posynomial} functions of a parameter vector $\pi$. In the latter case, minimizing spectral radius under posynomial constraints in $\pi$ can be solved in polynomial time via the Collatz-Wielandt formula \cite{Friedland81} and geometric programming \cite[p.\,165]{Boyd04}. (Note that in our \minradius problem the simplex constraint $\sum_{i=1}^n \pi_i = 1$ is not a posynomial constraint; it would be if the equality were replaced by the inequality `$\leq$'.)

\subsection{Stability}

\begin{problem}[\maxradius]
Given a set of real matrices, is there a convex combination of those whose spectral radius is larger than one? 
\end{problem}

\begin{theorem}
The \maxradius problem is NP-hard.
\label{maxr}
\end{theorem}

\begin{proof}
We reduce from the \clique problem. This problem asks, for a given undirected graph $\G=(V,E)$ and a positive integer $j \leq |V|$, whether $\G$ contains a clique $V'$ (a set of pairwise adjacent vertices) of size $j$ or more. This problem is NP-complete~\cite{Garey79}. 

Let $(\G,j)$ be an instance of \clique, \rmjn{with $j \geq 2$}. The reduction follows the pattern of the proof of Theorem \ref{minr}, with $C$, $c_i$, and $e_i$ as defined therein. We construct $n$ block matrices $A_i$, for $i=1,\ldots,n$:
 \begin{equation}
 A_i = 
 \begin{bmatrix} 
  \emptyset & c_i \\
  e_i^\top  & r 
 \end{bmatrix},
 \label{eq-A2}
\end{equation}
\rmjn{where $r=\frac{1}{2}+\frac{1}{2(j-1)}$}.
Their convex combination $\B = \sum_{i=1}^k \pi_i A_i$  reads
 \begin{equation}
 \B = 
 \begin{bmatrix} 
  \emptyset & C\pi \\
  \pi^\top  & r
 \end{bmatrix} .
 \label{eq-B2}
\end{equation}
If $\omega(\G)$ is the size of the largest clique of $\G$, the following holds \cite{Motzkin65}:
\begin{equation}
  \frac{1}{2}-\frac{1}{2\omega(\G)} = \max_{y \in \simplex} \, y^\top C \, y \, .
\label{eq-motzkin}
\end{equation}
Hence, in analogy with the proof of Theorem 1, the existence of a vector $\pi \in \simplex$ that satisfies 
the spectral radius condition  $\max(|\sigma(\B)|) > 1$ would imply $\pi^\top C \pi > 1-r$, and hence $\omega(\G)>\rmjn{j-1}$, which completes the reduction.
\end{proof}

\section{Continuous-time stabilizability with fixed uncertainty}

In this section we show that the problem of testing if a continuous-time system with fixed uncertainty is stabilizable is NP-hard. We prove this as a corollary of a more general result involving M-matrices, which in turn follows as a corollary of Theorem \ref{minr}. 
Recall that a Z-matrix is a real matrix with nonpositive off-diagonal entries, and a nonsingular M-matrix is a Z-matrix that is positive stable, that is, all its eigenvalues have positive real part \cite{Berman94}.  \rmjn{Note that, if $A$ is a nonsingular M-matrix, then $-A$ is a Metzler and Hurwitz matrix.  Such matrices have attracted some attention in recent years in the switching systems literature, because their algebraic properties make them easier to analyze than general matrices (see \cite{fainshil2009stability,valcher-positive,jungers-invariant}).}

\begin{problem}[The \mmatrix problem]
Given a set of real matrices, is there a convex combination of those that is a nonsingular M-matrix? 
\end{problem}

\begin{theorem}
The \mmatrix problem is NP-hard.
\label{np}
\end{theorem}

\begin{proof}
We construct a matrix $M_\pi = I - \B$, where the matrix $\B$ is as defined in \eq{B} \rmjn{(with $r=1-\frac{1}{j-1}$)}. Note that $M_\pi$ is the convex combination of the Z-matrices
\begin{equation}
A_i = 
 \begin{bmatrix} 
  I & -(e_i+c_i) \\
  -e_i^\top  & \frac{1}{\rmjn{j-1}}
 \end{bmatrix}, \quad  i=1,\ldots,n.
 \label{eq-A3}
\end{equation} 
The result then follows from the fact that $M_\pi$ is a nonsingular M-matrix if and only if the spectral radius of $\B$ is less than one \cite[p.\,133, definition 1.2]{Berman94}.
\end{proof}

We note that there exist more than 50 equivalent conditions for a Z-matrix to be a nonsingular M-matrix \cite{Berman94}. Hence, Theorem \ref{np} establishes NP-hardness of all `polytopic' versions of them. For instance, using conditions D$_{16}$ and N$_{38}$ in \cite[p.\,135,137]{Berman94}, we directly establish NP-hardness of the problems of testing if there exists a convex combination of a given set of matrices, whose real eigenvalues are all positive (D$_{16}$) or whose inverse is a nonnegative matrix (N$_{38}$).

NP-hardness of testing the stabilizability of a continuous-time system with fixed uncertainty now follows as a simple corollary of Theorem \ref{np}.
\rmj{
\begin{corollary}\label{cor-fcstabilizability}
Stabilizability in continuous time with fixed uncertainty is NP-hard.  That is, given a polytope of matrices, it is NP-hard to decide whether there exists a Hurwitz matrix in the polytope.
\end{corollary}
\begin{proof}
 If $M_\pi$ is a nonsingular M-matrix, then $-M_\pi$ is Hurwitz. All the matrices in the polytope whose vertices are defined in \eq{A3} are $Z$-matrices, and thus, requiring the existence of a M-matrix is equivalent to require the existence of a matrix whose eigenvalues have positive real part.  It is thus sufficient to take the same matrices as in the construction of the proof of Theorem \ref{np} multiplied by $-1$ in order to finish the proof:
 \begin{equation}
A_i = 
 \begin{bmatrix} 
  -I & e_i+c_i \\
  e_i^\top  & -\frac{1}{\rmjn{j-1}} 
 \end{bmatrix}, \quad  i=1,\ldots,n.
 \label{eq-fcstabilizability}
\end{equation} 
\end{proof}
}

\section{Stabilizability of continuous-time systems with switching uncertainty}\label{sec-scs}

To the best of our knowledge, this problem has not been studied in the literature, and we have not been able to settle it either.  The decidability problem seems as hard as the other ones considered above (if not harder), and is most probably NP-hard.  The simplest way of proving it would be to reuse the constructions used for the proof of the fixed uncertainty continuous-time stabilizability problem.  Indeed, we conjecture that the sets of matrices constructed in \eq{fcstabilizability} are \emph{switching} stabilizable if and only if they are \emph{fixed uncertainty}-stabilizable.  It is known that such a situation does not always hold, i.e., that there are sets of matrices that are continuous time stabilizable with switched uncertainty, but such that, if the uncertainty is fixed, all the matrices in the set are unstable (see \cite{gurvits-mtns}).  Nevertheless, \rmj{after running some numerical experiments,} we conjecture that for the particular sets at stake in construction \eq{fcstabilizability}, fixed-uncertainty and switched uncertainty are the same, meaning that one cannot improve stability by switching from time to time among matrices in the set.  This would then imply that the problem is NP-hard, since it would be equivalent to decide fixed-uncertainty stabilizability on the matrices in \eq{fcstabilizability}, which we have shown to be NP-hard.  
\rmj{
\begin{conj}
The problem of deciding the stabilizability of a continuous-time system with switching uncertainty is NP-hard. Moreover, the construction in Corollary \ref{cor-fcstabilizability} is a valid reduction, that is, the set of matrices in \eq{fcstabilizability} are fixed-uncertainty stabilizable if and only if they are switched-uncertainty stabilizable.
\end{conj}}

\rmjn{Proving this conjecture would settle the last open question in Table \ref{table-recap}, establishing NP-hardness of the problem. However, it seems that classical tools from switching systems theory (like for instance, extremal norms or lifting techniques, see \cite{GP11,jungers-invariant,jungers_lncis}) are not sufficient to prove the conjecture, and new tools should be developed.  Moreover, the links with graph theory (independent sets and cliques) seem to make it a rich research topic, and we believe that it could lead to important further studies.}

\section{Conclusion}
In this paper we have surveyed the stability and stabilizability problem of (fixed or switched) uncertain linear systems from a computational complexity point of view.
We restricted ourselves to uncertainty sets described as polytopes.  Of course, one could think about any other kind of uncertainty.  In particular, we did not cover the case where the polytope is given by a list of linear inequalities, instead of a list of its vertices.  Since there may be an exponential gap between the number of facets and the number of vertices of a polytope (and conversely), it is not clear whether the NP-hardness results presented in this paper extend to this latter case.  This problem, though relevant in several situations, seems to have been less studied in the literature.  To the best of our knowledge, only the cases of continuous-time fixed uncertainty are known, both for the stability \cite{Nemirovskii_interval_matrix_NPhard} and for the stabilizability \cite{blondel1997np}.  Unsurprisingly, the results are also negative in this setting: in these papers, both questions are proved to be NP-hard.

In our opinion, the NP-hardness results presented here should not end the study of these problems, which are of importance in many situations in engineering.  Rather, we think that they should help us to delineate the feasible problems, in order to design the best possible algorithms to be used in practice to analyze or design dynamical systems of the type considered in the paper.

\bibliographystyle{plain}

\end{document}